\documentclass[a4paper, twoside]{scrartcl}
\usepackage{amsfonts,amssymb,amsmath,amsthm,amstext,amssymb,mathtools}
\usepackage{caption,subcaption,float,color,xcolor,graphicx,enumerate,hyperref}
\usepackage{dsfont}  
\usepackage[authoryear,sort,round]{natbib}  
\usepackage[noabbrev,capitalise,nosort]{cleveref}

\usepackage{tikz}
\usetikzlibrary{bayesnet,er,trees,shapes.symbols,mindmap,arrows,decorations,shapes.misc,shapes.arrows,chains,matrix,positioning,scopes,decorations.pathmorphing,patterns}
\usepackage{tikz-cd} 

\usepackage{array}
\newcolumntype{x}[1]{>{\centering\arraybackslash\hspace{0pt}}p{#1}}
\usepackage{booktabs,colortbl,xcolor,xfrac}
\usepackage{nicefrac}
\def\R{\mathbb{R}}
\def\N{\mathbb{N}}

\def\E{\mathbb{E}}

\DeclareMathOperator*{\Bin}{Bin}
\DeclareMathOperator*{\MSE}{MSE}

\DeclareMathOperator*{\argmax}{arg\,max}

\newcommand*{\dkl}[2]{D_{\textup{KL}} \left( #1 \parallel #2 \right) }
\newcommand*{\Cdot}{\setbox0\hbox{$x$}\hbox to\wd0{\hss$\cdot$\hss}}

\newcommand*{\defeq}{\coloneqq}

\DeclarePairedDelimiter\abs{\lvert}{\rvert}%
\DeclarePairedDelimiter\norm{\lVert}{\rVert}%
\makeatletter
\let\oldabs\abs
\def\abs{\@ifstar{\oldabs}{\oldabs*}}
\let\oldnorm\norm
\def\norm{\@ifstar{\oldnorm}{\oldnorm*}}
\makeatother

\definecolor{darkgreen}{rgb}{0,0.4,0}

\newcommand*{\meas}{x} 
\newcommand*{\meassp}{\mathcal{X}} 
\newcommand*{\measdim}{n} 
\newcommand*{\data}{X} 
\newcommand*{\param}{\theta} 
\newcommand*{\parsp}{\Theta} 
\newcommand*{\parden}{\pi}
\newcommand*{\pardensp}{\cP}
\newcommand*{\model}{\mathcal{M}} 
\newcommand*{\replications}{k} 

\newcommand*{\cN}{\mathcal{N}}
\newcommand*{\cI}{\mathcal{I}}
\newcommand*{\cP}{\mathcal{P}}

\newcommand*{\MMI}{\mathrm{MMI}}
\newcommand*{\MMIL}{\mathrm{MMIL}}
\newcommand*{\MLE}{\mathrm{MLE}}
\newcommand*{\MPLE}{\mathrm{MPLE}}
\newcommand*{\MAP}{\mathrm{MAP}}
\newcommand*{\true}{\text{true}}

\theoremstyle{plain}
\newtheorem{theorem}{Theorem}[section]
\newtheorem{proposition}[theorem]{Proposition}

\newtheorem{corollary}[theorem]{Corollary}
\theoremstyle{definition}
\newtheorem{definition}[theorem]{Definition}

\newtheorem{algorithm}[theorem]{Algorithm}

\newtheorem{example}[theorem]{Example}
\newtheorem{condition}[theorem]{Condition}
\newtheorem{remark}[theorem]{Remark}

\numberwithin{equation}{section}
\numberwithin{figure}{section}
\numberwithin{table}{section}

\usepackage{etex}
\usepackage[a4paper, top=88pt, bottom=88pt, left=72pt, right=72pt, headsep=16pt, footskip=28pt]{geometry}
\usepackage{hyperref}
\hypersetup{
	colorlinks,
	linkcolor=blue,
	citecolor=blue,
	urlcolor=blue,
}
\newcommand*{\email}[1]{\bgroup\color{blue}\href{mailto:#1}{#1}\egroup}
\renewcommand*{\url}[1]{\bgroup\color{blue}\href{#1}{#1}\egroup}
\usepackage{enumitem}
\setlist[enumerate]{nosep}
\setlist[itemize]{nosep}
\usepackage{mleftright} \mleftright
\renewcommand{\qedsymbol}{$\blacksquare$}
\renewenvironment{proof}[1][\proofname]{\noindent{\bfseries #1.} }{\hfill \qedsymbol \medskip}
\usepackage{scrlayer-scrpage, xhfill}
\automark[section]{section}
\setkomafont{pageheadfoot}{\normalcolor\sffamily}
\setkomafont{pagenumber}{\normalfont\normalsize\sffamily}
\clearpairofpagestyles
\let\oldtitle\title
\renewcommand{\title}[1]{\oldtitle{#1}\newcommand{\theshorttitle}{#1}}
\newcommand{\shorttitle}[1]{\renewcommand{\theshorttitle}{#1}}
\let\oldauthor\author
\renewcommand{\author}[1]{\oldauthor{#1}\newcommand{\theshortauthor}{#1}}
\newcommand{\shortauthor}[1]{\renewcommand{\theshortauthor}{#1}}
\cohead{\xrfill[0.525ex]{0.6pt}~\theshorttitle~\xrfill[0.525ex]{0.6pt}}
\cehead{\xrfill[0.525ex]{0.6pt}~\theshortauthor~\xrfill[0.525ex]{0.6pt}}
\newcommand*{\affilref}[1]{\ref{affiliation#1}}
\newcommand*{\affiliation}[3]{
	\footnotetext[#1]{\label{affiliation#2} #3}
}

\title{Objective Priors in the Empirical Bayes Framework}
\shorttitle{Objective Priors in the Empirical Bayes Framework}
\author{%
	Ilja~Klebanov\textsuperscript{\affilref{ZIB}}%
	\and
	Alexander~Sikorski\textsuperscript{\affilref{ZIB}}%
	\and
	Christof~Sch{\"u}tte\textsuperscript{\affilref{ZIB},\affilref{FUB}}%
	\and
	Susanna~R{\"o}blitz\textsuperscript{\affilref{UIB}}%
}
\shortauthor{I.~Klebanov, A.~Sikorski, C.~Sch{\"u}tte and S.~R{\"o}blitz}

\begin{document}
\maketitle
\affiliation{1}{ZIB}{Zuse Institute Berlin, \email{klebanov@zib.de}, \email{sikorski@zib.de}, \email{schuette@zib.de}}
\affiliation{2}{FUB}{Freie Universit{\"a}t Berlin, Department of Mathematics, \email{schuette@mi.fu-berlin.de}}
\affiliation{3}{UIB}{University of Bergen, Department of Informatics \email{susanna.roblitz@uib.no}}

\begin{abstract}
	\noindent\textbf{\textsf{Abstract:}}
	When dealing with Bayesian inference the choice of the prior often remains a debatable question. Empirical Bayes methods offer a data-driven solution to this problem by estimating the prior itself from an ensemble of data. 
In the nonparametric case, the maximum likelihood estimate is known to overfit the data, an issue that is commonly tackled by regularization. However, the majority of regularizations are ad hoc choices which lack invariance under reparametrization of the model and result in inconsistent estimates for equivalent models. 
We introduce a non-parametric, transformation invariant estimator for the prior distribution. Being defined in terms of the missing information similar to the reference prior, it can be seen as an extension of the latter to the data-driven setting.
This implies a natural interpretation as a trade-off between choosing the least informative prior and incorporating the information provided by the data, a symbiosis between the objective and empirical Bayes methodologies.
	
	\medskip

	\noindent\textbf{\textsf{Keywords:}}
	{Parameter estimation},
{Bayesian inference},
{Bayesian hierarchical modeling},
{invariance},
{hyperprior},
{MPLE},
{reference prior},
{Jeffreys prior},
{missing information},
{expected information}

	\medskip

	\noindent\textbf{\textsf{2010 Mathematics Subject Classification:}}
	62C12, 62G08
\end{abstract}

\section{Introduction}
\label{section:introduction}

Inferring a parameter $\param\in\parsp$ from a measurement $\meas\in\meassp$  using Bayes' rule
requires prior knowledge about $\param$, which is not given in many applications.
This has led to a lot of controversy in the statistical community and to harsh
criticism concerning the objectivity of the Bayesian approach.
In the past decades, Bayesian statisticians have given (at least) two solutions to the dilemma of missing prior information:
\begin{enumerate}[label=(\Alph*)]
\item \textbf{(non-informative) objective priors:} 
Objective Bayesian analysis and reference priors in particular \citep{bernardo1979reference,berger1992development,berger2009formal} apply mostly information theoretic ideas to construct priors that are invariant under reparametrization and can be argued to be non-informative.
\item \textbf{empirical Bayes methods:}
If independent measurements $\meas_m\in\meassp$, $m=1,\dots,M$, are given for a large number $M$ of `individuals' with individual parametrizations $\param_m\in\parsp$, which is the case in many statistical studies, empirical Bayes methods \citep{robbins1956,casella1985introduction,efron2012large,carlin2010bayes,maritz2018empirical} use this knowledge to construct an \emph{informative} prior as a first step and then apply it for the Bayesian inference of the individual parametrizations $\param_m$ (or any future parametrization $\param_\ast$ with measurement $\meas_\ast$) in a second step.
A typical application is the retrieval of patient-specific
parametrizations in large clinical studies, see e.g.\ \citet{neal2010current}.
However, as discussed below, many of these methods fail to be consistent under reparametrization.
\end{enumerate}
The aim of this manuscript is to extend the construction of reference priors to the 
empirical Bayes framework in order to derive transformation invariant and informative priors from such `cohort data'.
We will perform this construction along the lines of the definition of reference priors in \citet{berger2009formal} and use a similar notation. Likewise, we concentrate mainly on problems with one continuous parameter ($\parsp\subseteq\R^d,\ d=1$), possible generalizations to the multiparameter case $d>1$ are addressed shortly in \Cref{section:HighDimensions}.

This paper is organized as follows.
After introducing the notation in \Cref{section:Setup}, we discuss empirical Bayes methods and the inconsistency of maximum penalized likelihood estimation (MPLE) under reparametrization, which is the main motivation for our work, in \Cref{section:EmpiricalBayesMethods}.
\Cref{section:ObjectiveEmpirical} provides a solution to this issue by following the same ideas as in the construction of reference priors, resulting in a prior estimate we term the \emph{empirical reference prior}. A rigorous definition and analysis of empirical reference priors is given.
\Cref{section:Practical} provides two algorithms how the empirical reference prior can be implemented in practice (under the assumption of asymptotic normality) --- one based on optimization and the other on a fixed point iteration.
In \Cref{section:Numerical} we then apply our methodology to a synthetic data set, illustrating invariance of the empirical reference prior under reparametrizations, as well as to the famous baseball data set of \citet{efron1975data}, where we compare our approach to the James--Stein estimator.
In \Cref{section:HighDimensions} we discuss possible generalizations of our approach to the multiparameter case $d>1$, followed by a short conclusion in \Cref{section:Conclusion}.
\section{Setup and Notation}
\label{section:Setup}

We will work in the empirical Bayes framework described above and visualized in Figure \ref{fig:GraphicalModel}.
Specifically, denoting the likelihood model by
\[
\model = \big\{ p(\meas|\param),\,  \meas\in\meassp,\, \param\in\parsp\big\},
\qquad
\meassp\subseteq\R^\measdim,
\qquad
\parsp\subseteq\R^d,\ d=1,
\]
the data $\data = (\meas_1,\dots,\meas_M)$ is generated by a two-stage process, where the ``individual'' parametrizations $\param_1,\dots,\param_M$ are independent draws from the unknown prior distribution $p(\param)$ and each data point $\meas_m\in\meassp$ is drawn independently from $p(\meas|\param_m)$, $m=1,\dots,M$.
Note that $\model$ denotes the model corresponding to the complete observation \emph{vector} $\meas\in\meassp \subseteq\R^\measdim$ and that the data $\data = (\meas_1,\dots,\meas_M)$ consists of $M$ observation vectors $x_1,\dots,x_M\in\meassp$.
This convention is necessary because our theory requires the introduction of
(artificial) independent replications of the entire experiment, denoted by the model
\[
\model^\replications = \Big\{ p(\vec\meas|\param) = \prod_{i=1}^{\replications} p(\meas^{(i)}|\param) ,\,  \vec\meas = (\meas^{(1)},\dots,\meas^{(\replications)})\in\meassp^\replications,\,  \param\in\parsp\Big\}.
\]
We adopt the standard abuse of notation, denoting all density functions by the letter $p$ and letting the argument indicate which random variable it belongs to, e.g.\ $p(\meas)$ is the marginal density of $\meas$ while $p(\param|\meas)$ denotes the posterior density of $\param$ given $\meas$. In addition, $\parden(\theta)$ denotes any other possible (proposed, guessed or estimated) prior on $\parsp$ from some class $\pardensp$ of admissible priors, $p(\meas|\parden) := \int p(\meas|\param) \parden(\param) \mathrm d\param$ the corresponding prior predictive distribution and $\parden(\theta|\meas)\propto \parden(\theta)p(\meas|\param)$ the corresponding posterior.
Further, $\parden_{\true}(\param) \coloneqq p(\param)$ denotes the ``true'' data-generating prior.

The prior may be viewed as a hyperparameter $\parden\in\pardensp$, in which case we assume conditional independence of $\parden$ and $\meas$ given $\param$. The marginal likelihood of $\parden$ given the entire data $\data = (\meas_1,\dots,\meas_M)$ is given by
\begin{equation}
\label{equ:MarginalLikelihood}
L(\parden) = L(\parden\, |\, \model,\data) = \prod_{m=1}^{M}p(\meas_m|\parden),
\qquad
p(\meas|\parden) = \int p(\meas|\param)\, \parden(\param)\, \mathrm d\param .
\end{equation}
We will also assume both the parametric and the hyperparametric model to be
identifiable, see \cite[Section 5.5]{van2000asymptotic}, i.e.
\begin{align}
\label{equ:ParIdentifiability}
p(\meas|\param_1) = p(\meas|\param_2)
\quad &\Longleftrightarrow\quad 
\param_1 = \param_2,
\qquad\ 
\param_1,\param_2\in\parsp,
\\
\label{equ:PardenIdentifiability}
p(\meas|\parden_1) = p(\meas|\parden_2)
\quad &\Longleftrightarrow\quad 
\parden_1 = \parden_2,
\qquad
\parden_1,\parden_2\in\pardensp,
\end{align}
since otherwise there would be no chance to recover the true distribution $\parden_{\true}(\param) = p(\param)$ from no matter how many measurements.

\begin{figure} 
     \centering
        \begin{subfigure}[b]{0.28\textwidth}
                \centering
                \tikz{ %
                    \node[latent] (pi) {$p(\param)$} ; %
                    \node[latent, below=of pi] (theta) {$\param$} ; %
                    \node[obs, below=of theta] (X) {$\meas$} ; %
                    \plate[inner sep=0.25cm, yshift=0.12cm] {plate2} {(theta) (X)} {$M$}; %
                    \edge {pi} {theta} ; %
                    \edge {theta} {X} ; %
                  }
        \end{subfigure}               
        \hspace{0.9cm}     
     \begin{subfigure}[b]{0.48\textwidth}
             \centering
             \tikz{ %
             \node[latent] (pi) {$p(\param)$} ; %
             \node[below=1.3cm of pi] (dots1) {$\cdots$} ; %
             \node[latent, left=of dots1] (theta1) {$\theta_1$} ; %
             \node[obs, below=1.3cm of theta1] (X1) {$\meas_1$} ; %
             \node[latent, right=of dots1] (thetaM) {$\theta_M$} ; %
             \node[obs, below=1.3cm of thetaM] (XM) {$\meas_M$} ; %
             \node[right=of X1] (dots2) {$\cdots$} ; %
             	\path
                 (pi) edge[->] node [right = 1.4cm]{$\param_m\stackrel{\textup{i.i.d.}}{\sim} \parden_{\true}(\param) = p(\param)$} (thetaM)
                 (thetaM) edge[->] node [right = 0.5cm]{$\meas_m\sim p(\meas|\param_m)$} (XM)
                 (pi) edge[->] (theta1)
                 (theta1) edge[->] (X1);
               }
     \end{subfigure}
        \hspace{0.9cm}     
        \caption{Graphical model and schematic representation of the underlying probabilistic model. Such a setup will be referred to as empirical Bayes framework.}
        \label{fig:GraphicalModel}
\end{figure}
\section{Inconsistency of Empirical Bayes Methods}
\label{section:EmpiricalBayesMethods}
Roughly speaking, empirical Bayes methods perform statistical inference in two steps -- first, they estimate the prior from the entire data  $\data = (\meas_1,\dots,\meas_M)$, and second, they apply Bayes' rule with that prior for each $\meas_m$ separately\footnote{For theoretical purposes one should rather think of applying Bayes' rule to some future measurement $\meas_\ast$ to infer its parametrization $\param_\ast$ in order to avoid reusing the data. In practice and for a large number $M$ of measurements this distinction is of little relevance, since the influence of one data point $x_m$ on the prior estimation can usually be neglected.}.
We concentrate on the first step, which is often performed by maximizing the marginal likelihood $L(\parden)$, for which the EM algorithm introduced by \citet{dempster1977maximum} is the standard tool.
This procedure can be viewed as an interplay of frequentist and
Bayesian statistics: The prior is chosen by maximum likelihood estimation (MLE), the actual individual parametrizations are then inferred using Bayes' rule.
However, in the nonparametric case (meaning, that no finite-dimensional parametric form of the prior is assumed), it can be proven that the marginal likelihood $L(\parden)$ is maximized by a discrete distribution $\parden_\MLE$
\begin{equation}
\label{equ:badMLE}
\parden_{\MLE}
=
\argmax_{\parden}\, \log L(\parden)
=
\sum_{\nu=1}^{N} w_\nu\delta_{\check\param_\nu},
\qquad
N\le M,\ w_\nu>0,\ \sum_{\nu=1}^{N} w_\nu = 1,
\end{equation}
with at most $M$ nodes $\check\param_\nu$, see \citet[Theorems 2-5]{laird1978nonparametric} or \citet[Theorem 21]{lindsay1995mixture}.
This typical issue of \emph{overfitting the data} is often dealt with by subtracting a \emph{roughness penalty} (or \emph{regularization term}) $\Phi(\parden) = \Phi(\parden\, |\, \model)$ from the marginal log-likelihood function $\log L(\parden)$, such that smooth or non-informative priors are favored, resulting in the so-called maximum penalized likelihood estimate (MPLE):
\begin{equation}
\label{equ:defMPLE}
\parden_{\MPLE} = \argmax_{\parden}\, \log L(\parden) - \gamma\Phi(\parden).
\end{equation}
The constant $\gamma>0$ balances the trade-off between goodness of fit and smoothness or non-informativity of the prior. This approach can be viewed from the Bayesian perspective as choosing a hyperprior $p(\parden) \, \propto\, e^{-\gamma\Phi(\parden)}$ for the hyperparameter $\parden$ and performing a maximum a posteriori (MAP) estimation for $\parden$:
\begin{equation}
\label{equ:MAPgleichMPLE}
\parden_\MAP = \argmax_{\parden}\, L(\parden)\, p(\parden)
=
\argmax_{\parden}\, \log L(\parden)-\gamma\Phi(\parden) = \parden_\MPLE\, .
\end{equation}
Favorable properties of the roughness penalty function  $\Phi(\parden) = \Phi(\parden\, |\, \model)$ are:
\begin{enumerate}[label=(\alph*)]
  \item non-informativity: without any extra information about the parameter
  or the prior, we want to keep our assumptions to a minimum (in the
  sense of objective Bayes methods);
  \item invariance under transformations of the parameter space $\parsp$ (reparametrizations);
  \item invariance under transformations of the measurement space $\meassp$;
  \item convexity: Since $\log L(\parden)$ is concave in the NPMLE case \citep[Section 5.1.3]{lindsay1995mixture}, a convex penalty function $\Phi(\parden)$ would guarantee a concave optimization problem \eqref{equ:defMPLE};
  \item natural and intuitive justification.
\end{enumerate}
The penalty functions currently used are mostly ad hoc and rather
brute force solutions that confine amplitudes, e.g.\ ridge regression \citep[Section 1.6]{mclachlan2007algorithm}, or derivatives \citep{good1971nonparametric,silverman1982estimation} of the prior, which are neither invariant under reparametrizations nor have a natural derivation.
A more contemporary alternative is to use Dirichlet process hyperpriors $p(\pi)$, which have similar limitations.
In order to incorporate the notion of non-informativity, 
\citet{good1963maximum} suggested to use the entropy as a roughness
penalty,
\begin{equation}
\label{equ:EntropyPenalty}
\Phi_{H_\param}(\parden) = -\, H_\param(\parden) = \int_\meassp \parden(\param)\log \parden(\param)\,
\mathrm d\param,
\end{equation}
which is a natural approach from an information-theoretic point of view, since high entropy corresponds to high uncertainty or thereby non-informativity of the prior.
However, $\Phi_{H_\param}$ is not invariant under reparametrizations, making it, as Good puts it, ``somewhat arbitrary'' \citep[p. 912]{good1963maximum}:
\begin{quote}
	``It could be objected that, especially for a continuous distribution, entropy is
	somewhat arbitrary, since it is variant under a transformation of the
	independent variable.''
\end{quote}
Following Shannon's derivation of the entropy $H_\param$, we will explain why the
concepts of mutual information and missing information, both of which are invariant under transformations, are far more natural quantities to use in our setup.
Prior to that, let us make the notion of invariance more precise.

\subsection{Invariance under Transformations}
Invariance under transformations guarantees consistency of the resulting probability density estimate and follows the following principle:
If two statisticians use equivalent models to explain equivalent data their
results must be consistent, or, as \citet{shore1980axiomatic} put it,
\begin{quote}
	``[\ldots] reasonable methods of inductive inference should lead to consistent results when there are different ways of taking the same information into account (for example, in different coordinate systems).''
\end{quote}
Let us start with the definition:
\begin{definition}
\label{def:Invariance}
Let $\data=(\meas_1,\dots,\meas_M)$ be data generated by the hierarchical model described in \Cref{section:Setup} with $\model = \{ p(\meas|\param),\,  \meas\in\meassp,\, \param\in\parsp\subseteq\R^d\}$.
Further, let $F = F[\model,\data]$ be a function operating on probability densities in $\R^d$.
\begin{enumerate}[label=(\roman*)]
	\item We call $F$ \emph{invariant under transformations of the parameter $\param$} or \emph{invariant under reparamet-rization}, if $F[\model,\data](\parden) = F[\model_\varphi,\data](\varphi_\ast \parden)$ for any diffeomorphism $\varphi\colon \parsp\to\tilde \parsp,\, \param\mapsto\tilde \param$, where the transformed model $\model_\varphi$ is defined by
	\begin{equation}
	\label{equ:phiTransform}
	p(\meas|\tilde{\param}) = p\big(\meas|\param
	=
	\varphi^{-1}(\tilde{\param})\big).		
	\end{equation}
	\item We call $F$ \emph{invariant under transformations of the measurement $\meas$}, if $F[\model,\data] = F[\model_\psi,\data_\psi]$ for any diffeomorphism $\psi\colon \meassp\to\tilde \meassp,\ \meas\mapsto\tilde \meas$, where the transformed model $\model_\psi$ and data $\data_\psi = (\tilde \meas_1,\dots,\tilde \meas_M)$ are defined by
	\begin{equation}
	\label{equ:psiTransform}
	p(\tilde\meas|\param)
	=
	\psi_\ast p(\meas|\param)\Big|_{\tilde{\meas}}
	=
	\left|\det(\psi^{-1})'(\tilde x)\right| \cdot  p\big(x = \psi^{-1}(\tilde x)\, |\, \param\big),
	\qquad
	\tilde\meas_m = \psi(\meas_m).
	\end{equation}	
\end{enumerate}
Here and in the following, $\varphi_\ast$ and $\psi_\ast$ denote the pushforwards of some measure (or density) under $\varphi$ and $\psi$, respectively.
\end{definition}

If we wish to use the function $F$ for the estimation of the prior $\parden$ (as in equation \eqref{equ:defMPLE} for $F = \log L - \gamma\Phi$), invariance of $F$ causes the following diagrams to commute. Note that, if we restrict the considerations to some class $\pardensp$ of admissible priors, then this class needs to be transformed correspondingly,
$
\pardensp_\varphi := \{\varphi_\ast\parden,\ \parden\in\pardensp\},
$
such that the considered class of priors is consistent under reparametrization.
\begin{figure}[H]
	\centering
	\begin{tikzcd}[column sep=huge,row sep=huge]
		(\model,\pardensp,\data)
		\arrow[r,shift left=.6ex,"\varphi",mapsto]
		\arrow[d,"\text{estimation of } \parden" description,mapsto]
		&
		(\model_\varphi,\pardensp_\varphi,\data)
		\arrow[l,shift left=.6ex,"\varphi^{-1}",mapsto]
		\arrow[d,"\text{estimation of } \parden" description,mapsto]
		\\
		\parden_{\textup{est}}(\param)
		\arrow[r,shift left=.6ex,"\varphi_\ast",mapsto]
		&
		\parden^\varphi_{\textup{est}}(\tilde \param)
		\arrow[l,shift left=.6ex,"\varphi^{-1}_\ast",mapsto]
	\end{tikzcd}
	\hspace{0.8cm}
	\begin{tikzcd}[column sep=huge,row sep=huge]
		(\model,\pardensp,\data)
		\arrow[r,shift left=.6ex,"\psi",mapsto]
		\arrow[d,"\text{estimation of } \parden" description,mapsto]
		&
		(\model_\psi,\pardensp,\data_\psi)
		\arrow[l,shift left=.6ex,"\psi^{-1}",mapsto]
		\arrow[d,"\text{estimation of } \parden" description,mapsto]
		\\
		\parden_{\textup{est}}(\param)
		\arrow[r,shift left=.6ex,"\psi_\ast",mapsto]
		&
		\parden^\psi_{\textup{est}}(\param)
		\arrow[l,shift left=.6ex,"\psi^{-1}_\ast",mapsto]
	\end{tikzcd}
	\caption{Commutative diagrams illustrating the consistency of the density estimate $ \parden_{\textup{est}}(\param)$. If the estimation is performed in a transformed parameter space $\tilde \parsp = \varphi(\parsp)$ or measurement space $\tilde \meassp = \psi(\meassp)$ (the model $\model$, the class $\pardensp$ of admissible priors and the data $\data=(\meas_1,\dots,\meas_M)$ are transformed accordingly), the results should be consistent. Note that, in the empirical Bayes framework, the data $\data$ enters in the estimation of the prior, hence the commutativity of the right diagram is not trivial.}
	\label{fig:CommutativeDiagramInvariance}
\end{figure}
One class of functions that fulfills these invariance properties is introduced in the following theorem, where we also show that the marginal likelihood $L(\parden)$ is transformation invariant up to a constant.
\begin{theorem}
\label{theorem:invarianceLandI}
Let $\data=(\meas_1,\dots,\meas_M)$ be data stemming from a model $\model = \{ p(\meas|\param),\,  \meas\in\meassp,\, \param\in\parsp\}$ and
\[
F_g(\parden)
=
F_g[\model](\parden)
=
\int_\parsp \int_\meassp  \parden(\param)\,  p(\meas|\param)\,
g\Big(\frac{p(\meas|\param)}{ p(\meas|\parden)}\Big)
\mathrm d\meas\, \mathrm d\param
\]
for some measurable function $g\colon\R\to\R$, such that the above integral is well-defined.
Then $F_g$ is invariant under transformations of $\param$ and $\meas$.
Further, the marginal likelihood $L(\parden) = L(\parden\, |\, \model,\data)$ defined by \eqref{equ:MarginalLikelihood} is invariant under transformations of $\param$ and $\meas$ up to a multiplicative constant.
\end{theorem}

\begin{proof}
This is a straightforward application of the change of variables formula.
\end{proof}
\section{Objective Bayesian Approach to Empirical Bayes Methods}
\label{section:ObjectiveEmpirical}

The lack of invariance of common empirical Bayes methods described above will now be tackled by an approach similar to the construction of reference priors and performed along the lines of \citet{berger2009formal}.
The two key ingredients for defining reference priors are permissibility, which yields a rigorous justification for dealing with improper priors, and the Maximizing Missing Information (MMI) property, which is derived from information theoretic considerations and can be argued to guarantee the least informative prior.

Permissibility allows for the (positive and continuous) prior to be improper as long as it yields a proper posterior for each measurement (which is the object Bayesian statisticians are actually interested in) and these posteriors can be approximated by using proper priors arising from restricting the prior to compact subspaces $\parsp_i\subseteq\parsp$.

In the empirical Bayes framework, where the aim is to approximate the true (proper) prior $\parden_{\true}(\param) = p(\param)$, improper priors are less of an issue and we will limit ourselves to proper priors. Furthermore, it is completely unclear how to deal with improper priors in this framework, since the `restriction property' of reference priors is neither achievable nor desirable, see \Cref{rem:ResrictionProperty} and \Cref{ex:ResrictionProperty} below.
For this reason and since the concept of permissibility has been elaborated extensively in \citet{berger2009formal}, we will just state its definition.
\begin{definition} 
A strictly positive continuous function $\parden(\param)$ is a \emph{permissible prior} for model $\model$, if
\begin{enumerate}[label=(\roman*)]
\item
for each $\meas\in\meassp$, $\parden(\param|\meas)$ is proper, i.e. $\int_\parsp p(\meas|\param)\, \parden(\param)\, \mathrm d\param < \infty$,
\item for any increasing sequence of compact subsets $\parsp_i\subseteq\parsp$, $i\in\N$, with $\bigcup_i\parsp_i = \parsp$, the corresponding posterior sequence $\parden_i(\param|\meas)$ is expected logarithmically convergent to $\parden(\param|\meas)$,
\[
\lim_{i\to\infty} \dkl{\parden_i(\Cdot|x)}{\parden(\Cdot|x)} = 0,
\qquad
\parden_i = \frac{ \parden\, \mathds 1_{\parsp_i}}{\int_{\parsp_i} \parden(\param)\, \mathrm d\param} .
\]
\end{enumerate}
\end{definition}
Here, $\mathds 1_{\parsp_i}$ is the indicator function of the subset $\parsp_i \subseteq \parsp$ and $\dkl{\Cdot}{\Cdot}$ denotes the Kullback-Leibler divergence defined by
\[
\dkl{p}{q} \coloneqq \int_\parsp p(\param)\, \log\bigg(\frac{p(\param)}{q(\param)}\bigg)\, \mathrm d\param .
\]
While permissibility is rather a technicality for dealing with improper priors, the MMI property should be seen as the defining property of reference priors and will now be discussed in more detail.

As motivated in the introduction, penalizing by means of entropy provides a
natural approach to incorporate the idea of non-informativity about the
parameter into the inference process.
However, if we follow Shannon's derivation of the entropy $H_\param$, we see that it is not the appropriate notion in our setup.
\citet{shannon2001mathematical} derived the entropy $H_\param$ from the insight that the proper way to quantify the information gain, when an event with
probability $p\in [0,1]$ actually occurs, is $-\log(p)\in [0,\infty]$. He then defined the entropy as
the expected information gain. However, the continuous analogue to this notion,
the differential entropy $H_\param$ given by \eqref{equ:EntropyPenalty}, faces several complications:
\begin{itemize}
	\item
	The information gain $-\log(\parden(\param))$ from observing the value $\param$, as well as the entropy $H_\param$ itself can become negative, which is difficult to interpret.
	\item
	$H_\param$ is variant under transformations of $\param$, leading to an inconsistent notion of information.
	\item
	The information gain $-\log(\parden(\param))$ relies on a \emph{direct} and \emph{exact} (error-free) measurement of $\param$, which is not plausible in the continuous case.
\end{itemize}
The last point becomes even more relevant in the empirical Bayes framework, where $\param$ is not (and usually cannot be) measured directly, but is inferred from the measurement $\meas$
of another quantity. The appropriate notion for the information gain in this setup is the Kullback-Leibler divergence $\dkl{\parden(\Cdot|x)}{\parden}$ between posterior and prior \citep{zbMATH03065448}.
Its expected value (from one observation of model $\model$), the so-called \emph{expected information} or \emph{mutual information} of $\param$ and $\meas$, is always non-negative and invariant under transformations.
\begin{definition} 
The \emph{expected information} gained from one observation of model $\model$ on a parameter $\param$ with prior $\parden(\param)$ is given by
\[
\cI[\parden\, |\, \model]
=
\int_\meassp p(\meas|\parden)\, \dkl{\parden(\Cdot|x)}{\parden}\, \mathrm d\meas
=
\int_\parsp \int_\meassp  \parden(\param)\,  p(\meas|\param)\, \log\Big(\frac{p(\meas|\param)}{ p(\meas|\parden)}\Big)\, \mathrm d\meas\, \mathrm d\param.
\]
\end{definition}
The expected information has very appealing properties for a penalty term:
\begin{theorem}
\label{theorem:ExpectedInformationConcaveInvariant}
The expected information $\cI[\parden\, |\, \model]$ is concave in $\parden$ and invariant under transformations of $\param$ and $\meas$.
\end{theorem}
\begin{proof}
Concavity is proven in \citet[Theorem 2.7.4]{cover2012elements} while invariance follows directly from Theorem \ref{theorem:invarianceLandI}.
\end{proof}

As argued in \citet{bernardo1979reference}, the quantity $\cI[\parden\, |\, \model^\replications]$, the expected information on $\param$ gained from $\replications$ independent observations of $\model$, describes the \emph{missing information} on $\param$ as $\replications$ goes to infinity:
\begin{quote}
	``By performing infinite replications of $\model$ one would get to know precisely the value of $\param$. Thus, $\cI[\parden\, |\, \model^\infty]$ measures the amount of missing information about $\param$ when the prior is $\parden(\theta)$.
	It seems natural to define ``vague initial knowledge'' about $\param$ as that described by the density $\parden(\param)$ which maximizes the missing information in the class $\pardensp$.''
	\footnote{The notation in this quotation has been adapted to ours. Note that the first statement makes use of the identifiability \eqref{equ:ParIdentifiability} of $\param$.}
\end{quote}
Following this idea, maximizing the missing information ($\MMI$) results in the least informative prior, making $\Phi(\parden) = - \cI[\parden\, |\, \model^\infty]$ an appealing penalty term in \eqref{equ:defMPLE}.
It is now tempting to define empirical reference priors by
\begin{equation}
\label{equ:badDefinitionER}
\parden_{\ast} = \argmax_\parden\,  \big(\log L(\parden)  + \gamma\, \cI[\parden\, |\, \model^\infty]\big).
\end{equation}
However, since $\cI[\parden\, |\, \model^\replications]$ typically diverges for $\replications\to\infty$, the following detour around the optimization formulation \eqref{equ:badDefinitionER} appears necessary (as we will see in Section \ref{section:AsymptoticNormality}, some simplifications are possible under certain regularity conditions):
\begin{definition} 
\label{def:MMI}
Let $\model = \big\{ p(\meas|\param),\,  \meas\in\meassp,\, \param\in\parsp\big\}$ be a model, $\pardensp$ be a class of prior functions $\parden$ with $\int p(\meas|\param)\, \parden(\param)\, \mathrm d\param <\infty$ and $\data = (\meas_1,\dots,\meas_M)$ be the data consisting of $M$ independent samples from $p(\meas)$. The function $\parden\in\pardensp$ is said to have
\begin{enumerate}[label=(\roman*)]
\item
the $\MMI = \MMI(\model,\pardensp)$ property for model $\model$ given $\pardensp$ if, for any compact set $\parsp_0\subseteq\parsp$ and any $\tilde\parden\in\pardensp$,
\begin{equation}
\label{equ:MMIproperty}
\lim_{\replications\to\infty} \left(\cI[\parden_0\, |\, \model^\replications] - \cI[\tilde\parden_0\, |\, \model^\replications]\right) \ge 0,
\end{equation}
where $\parden_0$ and $\tilde\parden_0$ denote the (renormalized) restrictions of $\parden$ and $\tilde\parden$ to $\parsp_0$;
\item
the $\MMIL(\data) = \MMIL(\model,\pardensp,\data,\gamma)$ property for model $\model$ given $\data$, $\pardensp$ and $\gamma>0$ if $\parden$ is a proper probability density and if, for any proper probability density $\tilde\parden\in\pardensp$,
\begin{equation}
\label{equ:MMIXproperty}
\lim_{\replications\to\infty} \Big( \left(\log L(\parden) + \gamma\cI[\parden\, |\, \model^\replications]\right) - \left(\log L(\tilde\parden) + \gamma\cI[\tilde\parden\, |\, \model^\replications] \right) \Big) \ge 0.
\end{equation}
\end{enumerate}
\end{definition}
Here, $\MMIL$ stands for ``maximizing (a trade-off between) missing information and log-likelihood".
Both definitions are only useful if the expected informations in \eqref{equ:MMIproperty} and \eqref{equ:MMIXproperty} are finite. 
This can be guaranteed by restricting ourselves to some convenient class $\pardensp$ of admissible priors. Typically, one requires strict positivity and continuity of the priors as well as the existence of proper posteriors, see \citet[Section 3.3]{berger2009formal}, but different choices of $\pardensp$ are also thinkable.

\subsection{The Formal Definition of Empirical Reference Priors}
Similar to \citet{berger2009formal}, and in accordance with their definition of reference priors, we now define empirical reference priors, which constitute the main contribution of this manuscript:
\begin{definition} 
\label{def:EmpiricalReferencePriors}
A function $\parden_{\textup{ref}}(\param) = \parden_{\textup{ref}}(\param\, |\, \model, \pardensp)$ is a \emph{reference prior} for model $\model$ given prior class $\pardensp$, if it is permissible and has the $\MMI$ property.
A probability density $\parden_{\textup{erp}}(\param) = \parden_{\textup{erp}}(\param\, |\, \model,\pardensp,\data,\gamma)$ is an \emph{empirical reference prior} for model $\model$ given prior class $\pardensp$, data $\data = (\meas_1,\dots,\meas_M)$ and smoothing parameter $\gamma>0$, if it has the $\MMIL(\data)$ property.
\end{definition}
Let us now formulate and prove the key properties of empirical reference priors.
\begin{theorem}[Invariance of the empirical reference prior]
\label{theorem:InvarianceERP}
The empirical reference prior is invariant under transformations of $\param$ and $\meas$ in the following sense:
\begin{align*}
\parden_{\textup{erp}}(\param\, |\, \model_\varphi,\pardensp_\varphi,\data,\gamma)
&=
\varphi_\ast\parden_{\textup{erp}}(\param\, |\, \model,\pardensp,\data,\gamma),
\\
\parden_{\textup{erp}}(\param\, |\, \model_\psi,\pardensp,\data_\psi,\gamma)
&=
\parden_{\textup{erp}}(\param\, |\, \model,\pardensp,\data,\gamma),
\end{align*}
where we adopted the notation from the definition of invariance and $\pardensp_\varphi := \{\varphi_\ast\parden,\ \parden\in\pardensp\}$.
\end{theorem}
\begin{proof}
Since $\log L(\parden)$ is invariant under transformations of $\param$ and $\meas$ up to an additive constant by Theorem \ref{theorem:invarianceLandI}, this is a direct consequence of Theorems \ref{theorem:invarianceLandI} and \ref{theorem:ExpectedInformationConcaveInvariant}.
\end{proof}

\begin{theorem}[Compatibility with sufficient statistics]
If the model $\model = \big\{ p(\meas|\param),\,  \meas\in\meassp,\, \param\in\parsp\big\}$ has a sufficient statistic $t = t(x)\in\mathcal T$, then
\[
\parden_{\textup{erp}}(\param\, |\, \model,\pardensp,\data,\gamma)
=
\parden_{\textup{erp}}(\param\, |\, \model_t,\pardensp,T,\gamma),
\]
where $T = t(X)\in\mathcal T^M$ and $\model_t = \big\{ p(t|\param),\,  t\in\mathcal T,\, \param\in\parsp_0\big\}$ is the corresponding model in terms of $t$.
\end{theorem}

\begin{proof}
Since $t$ is a function of $x$ and a sufficient statistic for $\param$, we obtain
\[
p(\meas|\param)
=
p(\meas,t(\meas)|\param)
=
p(\meas|t(\meas),\param)\, p(t(\meas)|\param)
=
p(\meas|t(\meas))\, p(t(\meas)|\param).
\]
This implies that the marginal log-likelihoods $\log L(\pi)$ and $\log L_t(\pi)$ agree up to an additive constant, where $L_t(\parden) = \prod_{m=1}^{M} p(t_m|\pi),\ t_m := t(x_m)$, denotes the marginal likelihood in terms of $t$:
\[
\log L(\parden)
=
\sum_{m=1}^{M} \log\int p(\meas_m|\param)\, \parden(\param)\, \mathrm d\param
=
\sum_{m=1}^{M} \log p(\meas_m|t_m)
+
L_t(\parden).
\]
Since the expected information is also invariant under such transformations, $\cI[\parden\, |\, \model^\replications] = \cI[\parden\, |\, \model_t^\replications]$, see \citet[Theorem 5]{berger2009formal}, this proves the claim.
\end{proof}

\begin{remark}
	\label{rem:ResrictionProperty}
	Reference priors have the appealing property that their restrictions to any compact subset $\parsp_0$ coincide with the reference priors on $\parsp_0$, see \citet[Section 5]{berger2009formal}:
	\[
	\parden_{\textup{ref}}(\param\, |\, \model,\pardensp)\big|_{\parsp_0}
	=
	\parden_{\textup{ref}}(\param\, |\, \model_0,\pardensp_0),
	\quad
	\model_0 = \big\{ p(\meas|\param),\,  \meas\in\meassp,\, \param\in\parsp_0\big\},
	\quad 
	\pardensp_0 = \big\{ \parden\big|_{\parsp_0},\, \parden\in\pardensp \big\}.
	\]
	However, unlike for objective priors in the absence of data, this property is not desirable in the empirical Bayes framework, as explained in the example below, and usually will not be fulfilled by empirical reference priors. Therefore, a definition of $\MMIL(\data)$ using restrictions of possibly improper priors (as in the definition of $\MMI$) is not meaningful and we are forced to limit ourselves to proper priors.
	This limitation is not too restrictive since the aim of empirical Bayes methods is to approximate the true prior $\parden_{\true}(\param)$ and improper priors do not play a major role.
	For compact parameter spaces $\parsp$ (and in all other cases for which the reference prior turns out to be proper), empirical reference priors provide a meaningful generalization of reference priors, which then correspond to the case $M=0$, the absence of data $\data$.
\end{remark}

\begin{example}
	\label{ex:ResrictionProperty}
	Let the true data-generating prior be the uniform prior $\parden_{\true} \equiv \frac{1}{2}$ on $\parsp = [0,2]$, i.e.\ $\param\sim{\textup{Unif}}([0,2])$, and $\model$ be the location model given by $\meas|\param \sim \cN(\theta,0.5^2)$.
	For a `large' data set $\data$ consisting of $M = 100$ measurements, the empirical reference prior $\parden_{\textup{erp}}(\param\, |\, \model,\pardensp,\data)$ can be expected to provide a good approximation of $\parden_{\true}(\param)$, hence its restriction to $\parsp_0 = [0,1]$ will be approximately uniform, see Figure \ref{fig:restriction}. However, the empirical reference prior $\parden_{\textup{erp}}(\param\, |\, \model_0,\pardensp_0,\data)$ on $\parsp_0$ has to put much more weight on values close to $1$ in order to explain the many measurements $\meas_m$ which are larger than $1$.
	\\
	Hence, unlike for reference priors, the equality $\parden_{\textup{erp}}(\param\, |\, \model,\pardensp,\data)\big|_{\parsp_0} = \parden_{\textup{erp}}(\param\, |\, \model_0,\pardensp_0,\data)$ is neither fulfilled nor desirable.
	Of course, in practice, the parameter space $\parsp$ should agree with (or at least include) the domain of the true prior in order to be consistent with the data generating distribution.	
	\begin{figure}[H]
		\centering
		\includegraphics[width=0.6\textwidth]{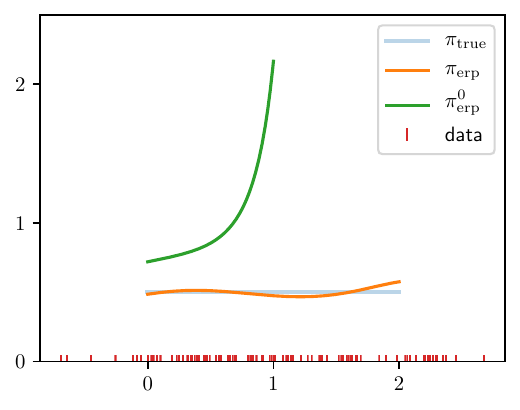}
		\caption{Inconsistency of the empirical reference prior under restrictions of the parameter space: $\parden_{\textup{erp}} := \parden_{\textup{erp}}(\param\, |\, \model,\pardensp,\data)\big|_{\parsp_0}$ restricted to $\parsp_0 = [0,1]$ does not agree with $\parden_{\textup{erp}}^0 := \parden_{\textup{erp}}(\param\, |\, \model_0,\pardensp_0,\data)$. Instead, $\parden_{\textup{erp}}^0$ puts a lot of weight close to the right boundary in order to explain all the data points larger than $1$. Unlike in the case of reference priors, this disagreement is reasonable in the presence of data. For simplicity, we chose $\gamma = 1$ in this example.}
		\label{fig:restriction}
	\end{figure}
\end{example}

\subsection{Choice of the Smoothing Parameter $\gamma$}
\label{section:ChoiceSmoothingParameter}
So far, it is completely unclear how the smoothing parameter $\gamma>0$ should be chosen. In fact, even entirely different ways of performing the trade-off between $L(\parden)$ and $\cI[\parden\, |\, \model^\infty]$ are thinkable. All theoretical results remain unchanged if we replace
\[
\log L(\parden) + \gamma\cI[\parden\, |\, \model^\replications]
\qquad
\text{by}
\qquad
\Psi(\log L(\parden),\cI[\parden\, |\, \model^\replications])
\]
in \eqref{equ:MMIXproperty}, where $\Psi\colon\R^2\to\R$ can be any concave function that is monotonically increasing in both arguments.
We will stick to the former formulation and choose $\gamma$ via likelihood cross-validation \citep[equation (3.43)]{silverman2018density},
\begin{equation}
\label{equ:CrossValidation}
\gamma_\ast
=
\argmax_{\gamma} \sum_{m=1}^{M} \log p\big( x_m\, |\, \parden_{\textup{erp}}(\param\, |\, \model, \pardensp, \data_{-m},\gamma) \big),
\end{equation}
where $X_{-m}\coloneqq \{\meas_{m'}\in \data \mid m'\neq m \}$ denotes the data set $X$ with the $m$-th point $x_m$ left out.
\begin{proposition}
\label{equ:InvarianceSmoothingParameter}
The smoothing parameter \eqref{equ:CrossValidation} is invariant under transformations of $\param$ and $\meas$, as long as the class $\pardensp$ of admissible priors is transformed accordingly, $\pardensp_\varphi = \{\varphi_\ast\parden,\ \parden\in\pardensp\}$.
\end{proposition}
\begin{proof}
From \eqref{equ:phiTransform}, \eqref{equ:psiTransform} and \Cref{theorem:InvarianceERP} we obtain
\begin{align*}
p\big( x\, |\, \parden_{\textup{erp}}(\tilde\param\, |\, \model_\varphi,\pardensp_\varphi,\data,\gamma) \big)
&=
p \big( x\, |\, \parden_{\textup{erp}}(\param\, |\, \model,\pardensp,\data,\gamma) \big),
\\
p \big( \tilde x\, |\, \parden_{\textup{erp}}(\param\, |\, \model_\psi,\pardensp,\data_\psi,\gamma) \big)
&=
C(\tilde x)\, p \big( x = \psi^{-1}(\tilde x)\, |\, \parden_{\textup{erp}}(\param\, |\, \model,\pardensp,\data,\gamma)\big),
\end{align*}
where we adopted the notation from the definition of invariance
and $C(x) = \left|\det(\psi^{-1})'(\tilde x)\right| > 0$ does not depend on $\gamma$.
This proves the claim.
\end{proof}
\subsection{Empirical Reference Priors under Asymptotic Normality}
\label{section:AsymptoticNormality}
As proven in \citet{clarke1994jeffreys}, the reference prior coincides with the Jeffreys prior \citep{jeffreys1946invariant,zbMATH03189754},
\begin{equation}
\label{equ:JeffreysPrior}
\parden^J(\param)\, \propto\, J(\param) := \sqrt{|\det i(\param)|},
\qquad
i(\param)
:=
\int_{\meassp} p(\meas|\param)\, \big(\nabla_\param\log p(\meas|\param)\big) \big(\nabla_\param\log p(\meas|\param)\big)^{\intercal}\, \mathrm d\meas,
\end{equation}
under certain regularity conditions, which basically ensure asymptotic posterior normality. Let us recall the basic results which we state for any dimension $d\in\N$ of the parameter space $\parsp$.
\begin{condition}
\label{cond:AsymptoticNormality}
The likelihood $p(\meas|\param)$ is twice continuously differentiable in $\param$ for almost every $\meas\in\meassp$. There exists $\epsilon>0$ and for every $\param$ there exists $\delta>0$ such that for all $j,k=1,\dots,d$ the functions
\[
\E\left[\abs{ \frac{\partial}{\partial\param_j} \log p(\meas|\param) }^{2+\epsilon}  \right]
\qquad
\text{and}
\qquad
\E\left[\sup_{\{\param : \|\param-\param'\|<\delta\}}\, \abs{ \frac{\partial^2}{\partial\param_j'\partial\param_k'}\log p(\meas|\param') }^2   \right]
\]
are finite and continuous in $\param$. The Fisher information matrix $i(\param)$ defined by \eqref{equ:JeffreysPrior} is positive definite for each $\param\in\parsp$ and equals $\tilde i(\param) = - \E\left[D_\param^2 \log p(\meas|\param)\right]$, where $D_\param^2$ denotes the Hessian matrix with respect to $\param$. The model $\model$ is identifiable as defined by \eqref{equ:ParIdentifiability}, the parameter space $\parsp$ is compact and $\pardensp$ consists of all positive and continuous probability density functions on $\parsp$ (this implies that $\pardensp$ is a convex set).
\end{condition}

\begin{proposition}
Under Condition \ref{cond:AsymptoticNormality}, there exist positive constants  $C_1,C_2>0$ such that, as $\replications\to\infty$,
\begin{equation}
\label{equ:AsymptoticMutualInformation}
\cI[\parden\, |\, \model^\replications]
=
C_1\log(C_2k) - \dkl{\parden}{J} + o(1),
\qquad
J(\param) = \sqrt{|\det i(\param)|}.
\end{equation}
\end{proposition}
\begin{proof}
See \citet{clarke1994jeffreys}.
\end{proof}

Since the first term on the right-hand side of \eqref{equ:AsymptoticMutualInformation} does not depend on the prior $\parden(\param)$ and the second one is independent of $k$,
the reference prior coincides with Jeffreys prior $\parden_J$ and
the definition of empirical reference priors recovers the form of the MPLE \eqref{equ:defMPLE}:

\begin{corollary}
Under Condition \ref{cond:AsymptoticNormality}, Jeffreys prior $\parden_J$ is the unique reference prior (up to scaling).
\end{corollary}
\begin{corollary}
Under Condition \ref{cond:AsymptoticNormality}, the empirical reference prior is given by the the following MPLE \eqref{equ:defMPLE},
\begin{equation}
\label{equ:KullbackLeiblerPenaltyMPLE}
\parden_{\textup{erp}} = \argmax_{\parden\in\pardensp}\, \log L(\parden) - \gamma \Phi_\cI(\parden),
\qquad
\Phi_\cI(\parden) := \dkl{\parden}{J},
\end{equation}
where the penalty term $\Phi_\cI$ will be referred to as the \emph{missing information penalty}.
\end{corollary}
Theorems \ref{theorem:invarianceLandI} and \ref{theorem:ExpectedInformationConcaveInvariant}
imply favorable properties of the missing information penalty $\Phi_\cI$ and, in particular, the existence of a unique empirical reference prior that is invariant under transformations of $\param$ and $\meas$.
\begin{corollary}
\label{corollary:StrictConvexityUniqueSolution}
Under Condition \ref{cond:AsymptoticNormality}, the optimization problem \eqref{equ:KullbackLeiblerPenaltyMPLE} is strictly concave in $\parden$ and invariant under transformations of $\param$ and $\meas$. Hence, since $\pardensp$ is convex, it has a unique solution $\parden_{\textup{erp}}$.
\end{corollary}
\begin{proof}
For the concavity of $\log L(\parden)$ see \citet[Section 5.1.3]{lindsay1995mixture}, while the strict convexity of $\Phi_{\cI}$ follows from \citet[Theorem 11]{van2014renyi}. Invariance is a direct consequence of Theorems \ref{theorem:invarianceLandI} and \ref{theorem:ExpectedInformationConcaveInvariant}.
\end{proof}
In order to realize the optimization formulation \eqref{equ:KullbackLeiblerPenaltyMPLE}, we now compute analytically the gradient of the functional to be maximized.
In addition, we characterize $\parden_{\textup{erp}}$ as the unique fixed point of a certain function $F^{\ast}\colon \pardensp \to \pardensp$, which motivates the fixed point iteration in \Cref{algorithm:FixedPointERP}.
\begin{theorem}
	\label{thm:FixpointFormulationForERP}
	Let \Cref{cond:AsymptoticNormality} hold. Then the gradient of the functional in \eqref{equ:KullbackLeiblerPenaltyMPLE},
	\begin{equation}
	\label{equ:PsiFunctionalToOptimize}
	\Psi(\parden)
	=
	\log L(\parden) - \gamma \Phi_\cI(\parden)
	=
	\sum_{m=1}^{M} \log  \int_{\parsp} p(\meas_{m}|\param)\, \parden(\param)\, \mathrm d\param  - \gamma \dkl{\parden}{J},
	\end{equation}
	with respect to $\langle \Cdot,\Cdot\rangle_{L^2(\parsp)}$ is given by
	\begin{equation}
	\label{equ:GradientPsi}
	(\nabla_{\parden} \Psi(\parden))(\param)
	=
	v^{\ast}(\param) - \int v^{\ast},
	\qquad
	v^{\ast}(\param)
	=	
	\sum_{m=1}^{M} \frac{p(\meas_{m}|\param)}{\int p(\meas_{m}|\param')\, \parden(\param')\, \mathrm d\param'} - \gamma \log \frac{\parden(\param)}{J(\param)}.
	\end{equation}	
	It follows that the empirical reference prior $\parden_{\textup{erp}}$ is the unique fixed point of $F^{\ast}\colon \pardensp \to \pardensp$,
	\begin{equation}
	\label{equ:FixpointFormulationForERP}
	F^{\ast}(\parden)
	\defeq
	\frac{F(\parden)}{\int (F(\parden))(\param)\, \mathrm d \param}\, ,
	\qquad
	F(\parden)
	\defeq
	J \, \exp\left(\gamma^{-1} \sum_{m=1}^{M} \frac{p(\meas_{m}|\Cdot)}{\int_{\parsp} p(\meas_{m}|\param')\, \parden(\param')\, \mathrm d \param' }\right).
	\end{equation}
\end{theorem}

\begin{proof}
	We consider perturbations of $\Psi(\parden)$ in arbitrary directions $v\in C(\parsp)$ with $\int v = 0$:
	\begin{align*}
		\frac{\mathrm d}{\mathrm d \varepsilon}\Psi(\parden + \varepsilon v)
		&=
		\sum_{m=1}^{M} \frac{\int p(\meas_{m}|\param)\, v(\param)\, \mathrm d\param}{\int p(\meas_{m}|\param')\, \parden(\param')\, \mathrm d\param'}
		\, -\, 
		\gamma \int v(\param)\, \log \frac{\parden(\param)}{J(\param)}\, \mathrm d\param
		\, -\, 
		\gamma \underbrace{\int v(\param)\, \mathrm d\param}_{=\, 0}
		\\
		&=
		\bigg\langle v \, ,\,  \sum_{m=1}^{M} \frac{p(\meas_{m}|\param)}{\int p(\meas_{m}|\param')\, \parden(\param')\, \mathrm d\param'} - \gamma \log \frac{\parden(\param)}{J(\param)} \bigg\rangle_{L^2(\parsp)},
	\end{align*}
	which proves \eqref{equ:GradientPsi}.
	The above inner product is zero for any $v\in C(\parsp)$ with $\int v = 0$ if and only if its second argument is constant in $\param$. Hence,	
	\[
	\parden_{\textup{erp}}(\param)
	\propto
	J(\param) \, \exp \bigg(\gamma^{-1} \sum_{m=1}^{M} \frac{p(\meas_{m}|\param)}{\int_{\parsp} p(\meas_{m}|\param')\, \parden(\param')\, \mathrm d \param' } \bigg),
	\]
	which proves the second claim. Note that $F^{\ast}\colon \pardensp \to \pardensp$, since we assumed the parameter space $\parsp$ to be compact, and that the uniqueness of the fixed point follows from the strict convexity of $\Psi$ (\Cref{corollary:StrictConvexityUniqueSolution}).
\end{proof}
\section{Practical Realization of the Empirical Reference Prior}
\label{section:Practical}

In this section we demonstrate how the empirical reference prior can be computed in the asymptotically normal case (i.e.\ under \Cref{cond:AsymptoticNormality}).
We suggest two algorithms, one of which is a straightforward optimization of the strictly concave functional in \eqref{equ:KullbackLeiblerPenaltyMPLE} and the other being a fixed point iteration, the convergence of which, however, is not yet fully understood.

For simplicity, we treat only the case where $\parsp = [a,b]$ is an interval and is discretized by the equidistant grid $G_K$ defined below.
Integrals over $\parsp$ will be approximated by the midpoint rule
\[
Q(f) \coloneqq \frac{b-a}{K} \sum_{k=1}^{K} f(\param_{k}),
\qquad
G_K = \Big\{ \param_{k}\defeq \frac{2k-1}{2K} \mid k=1,\dots,K \Big\}.
\]
Further, we assume that the Jeffreys prior $J(\param)$ can be evaluated pointwise in each $\param_{k}$, either because its analytic form is available or by a numerical approximation of the integral in \eqref{equ:JeffreysPrior}.
Finally, we consider the computation of $\parden_{\textup{erp}}(\param) = \parden_{\textup{erp}}(\param\, |\, \model, \pardensp, \data,\gamma)$ only for a given smoothing parameter $\gamma$.
The choice of $\gamma$ requires another optimization as discussed in \Cref{section:ChoiceSmoothingParameter}, this time in $\gamma$, where the algorithms below have to be executed many times for the computation of $\parden_{\textup{erp}}(\param\, |\, \model, \pardensp, \data_{-m},\gamma)$ with varying data sets $\data_{-m}$, see the likelihood cross-validation formula \eqref{equ:CrossValidation}.
While this ``brute force'' solution seems computationally challenging, the overall procedure could be performed in all our examples within just a few seconds (mainly due to the strict concavity of the optimization problem \eqref{equ:KullbackLeiblerPenaltyMPLE}).
A more elegant solution which combines the two optimization problems is imaginable, but goes beyond the scope of this paper. 
In the two algorithms below, we identify each prior $\parden$ with its discretized version $\hat{\parden} = (\parden(\param_{1}),\dots,\parden(\param_{K}))$ (and similar for $J$, $\parden_{t}$ etc.).
\begin{algorithm}
	\label{algorithm:OptimizationERP}	
	\ 
	\begin{enumerate}[label=(\roman*)]
		\item
		Formulate the discretized functional
		\[
		\hat{\Psi}(\parden)
		=
		\sum_{m=1}^{M} \log Q\big(p(\meas_m|\Cdot)\, \parden\big)
		- \gamma \, Q\big(\parden \, (\log \parden - \log J)\big)
		\]
		and its gradient in accordance with \eqref{equ:PsiFunctionalToOptimize} and \eqref{equ:GradientPsi}.		
		\item
		Optimize $\Psi(\parden)$ as a function of $\parden$ by an optimization algorithm of your choice (we use the method of moving asymptotes (MMA) of \citet{svanberg2002class} from the NLopt package provided by \citet{nlopt}). Renormalize $\hat{\parden}$ in each iteration step.
	\end{enumerate}
\end{algorithm}
Motivated by the fixed point characterization of $\parden_{\textup{erp}}$ in \Cref{thm:FixpointFormulationForERP}, it is tempting to implement the fixed point iteration $\parden_{t} = F^{\ast}(\parden_{t-1})$, $t\in\N$. However, our empirical studies showed that the resulting sequence $(\parden_{t})_{t\in\N}$ often fails to converge, indicating that $F^{\ast}$ is, in general, not a contraction.
This suggests to decelerate the iteration by choosing $\parden_{t} = (1-\tau_{t})\, \parden_{t-1} + \tau_{t} \,  F^{\ast}(\parden_{t-1})$ with step sizes $\tau_{t}\in [0,1]$.
While this gave satisfactory results in our examples, it is hard to characterize the step sizes for which the new sequence converges.
\begin{algorithm}
	\label{algorithm:FixedPointERP}
	Choose a threshold $\varepsilon>0$.
	\begin{enumerate}[label=(\roman*)]
		\item
		Choose an arbitrary positive probability density $\parden_{0}$, e.g.\ $\parden_{0}(\param) = (b-a)^{-1}$.
		\item
		For $t\in\N$, iterate:
		\begin{enumerate}[label=(\alph*)]
			\item compute
			$\quad \tilde \parden_{\textup{temp}}(\param)
			=
			\, J(\param) \, \exp\left(\gamma^{-1} \sum_{m=1}^{M} \frac{p(\meas_{m}|\param)}{Q\big(p(\meas_{m}|\Cdot)\, \parden_{t-1}\big)}\right)$;
			\item
			normalize: $\parden_{\textup{temp}} = \tilde\parden_{\textup{temp}} / Q(\tilde \parden_{\textup{temp}})$;
			\item choose a step size $\tau_{t}\in [0,1]$ and
			$\parden_{t} = (1-\tau_{t})\, \parden_{t-1} + \tau_{t} \, \parden_{\textup{temp}}$;
		\end{enumerate}
		until $Q(\abs{\parden_{\textup{temp}} - \parden_{t^{\ast}-1}}) < \varepsilon$ for some $t^{\ast}\in\N$.
		\item
		Set $\parden_{\textup{erp}} = \parden_{t^{\ast}}$.		
	\end{enumerate}
\end{algorithm}
\section{Numerical Computations}
\label{section:Numerical}

In the following, we apply the empirical reference prior approach to a synthetic data set as well as to the famous real-life data set of \citet{efron1975data}.
Our code implements \Cref{algorithm:OptimizationERP} and is available at \url{https://github.com/axsk/ObjectiveEmpiricalBayes.jl}
in the form of a Julia package.

\subsection{Example with synthetic data}
\label{section:NumericalSyntheticDataSet}
We illustrate the invariance under reparametrization using the location model
\begin{equation}
\label{equ:locModel}
\meas | \param \sim \mathcal{N}\left(\param, \sigma^2 \right),
\qquad
\param \sim \pi_{\true},
\end{equation}
with $\sigma = 0.3$ and $\pi_\text{true}$ being an equal mixture of $\mathcal{N}(1, 0.5^2)$ and $\mathcal{N}(3, 0.5^2)$, truncated to the interval $\parsp = [0,4]$.
We compare the performance of MPLE \eqref{equ:defMPLE} using Tikhonov regularization $\Phi= \left\| \cdot \right\|_{L^2}^{2}$, also known as ridge regression, and the empirical reference prior \eqref{equ:KullbackLeiblerPenaltyMPLE}, applied in both the untransformed space $\parsp = [0,4]$ as well as the space $\tilde \parsp$ transformed by $\varphi: \param \mapsto \tilde \param = \exp(\param)$.
\begin{figure}[H]
	\centering
	\begin{subfigure}[b]{0.44\textwidth}
		\centering	
		\includegraphics[width=\textwidth]{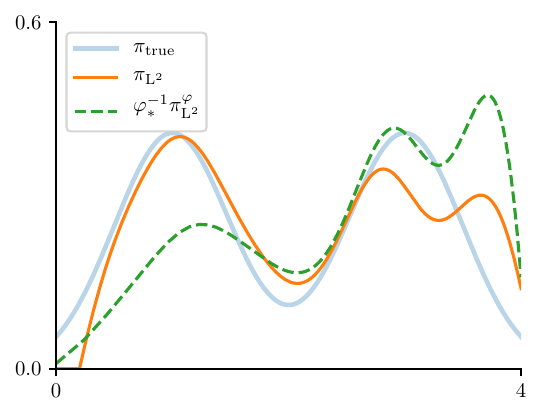}
	\end{subfigure}
	\hfill	
	\begin{subfigure}[b]{0.1\textwidth}		
		\centering
		\tikz{
			\draw [->] (-.8,3)  -- node[above] {$\varphi_\ast$}      (.8,3);
			\draw [->] (.8,2.5) -- node[below] {$\varphi^{-1}_\ast$} (-.8,2.5);
			\phantom{\draw [->] (-.8,0) -- (.8,0);}
		}
	\end{subfigure}
	\hfill
	\begin{subfigure}[b]{0.44\textwidth}
		\centering	                
		\includegraphics[width=\textwidth]{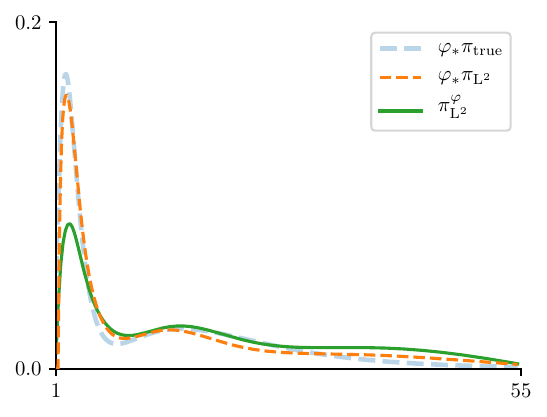}
	\end{subfigure}	
	\caption{Conventional penalty terms, here the Tikhonov-regularization, are variant under transformations $\varphi\colon \parsp\to\tilde\parsp$, resulting in inconsistent prior estimates (see also Figure \ref{fig:CommutativeDiagramInvariance}). If the estimation is performed in a transformed space $\tilde\parsp$ it gives a different estimate than the pushforward of the estimate in $\parsp$, $\parden^{\varphi}_{\textup{est}} \neq \varphi_\ast \parden_{\textup{est}}$. Dashed lines correspond to transformed densities.
		Here, $\pi_{L^2}$ denotes the MPLE using Tikhonov regularization as penalty. In order to demonstrate the lack of invariance of the density estimate, we chose the same smoothing parameter $\gamma$ in both spaces.
	}
	\label{fig:Tikhonov}
\end{figure}
\begin{figure}[H]
	\centering
	\begin{subfigure}[b]{0.44\textwidth}
		\centering	
		\includegraphics[width=\textwidth]{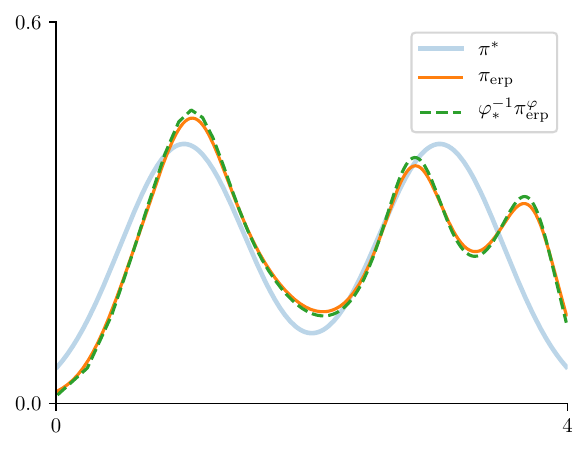}
	\end{subfigure}
	\hfill	
	\begin{subfigure}[b]{0.1\textwidth}		
		\centering
		\tikz{
			\draw [->] (-.8,3)  -- node[above] {$\varphi_\ast$}      (.8,3);
			\draw [->] (.8,2.5) -- node[below] {$\varphi^{-1}_\ast$} (-.8,2.5);
			\phantom{\draw [->] (-.8,0) -- (.8,0);}
		}
	\end{subfigure}
	\hfill
	\begin{subfigure}[b]{0.44\textwidth}
		\centering	                
		\includegraphics[width=\textwidth]{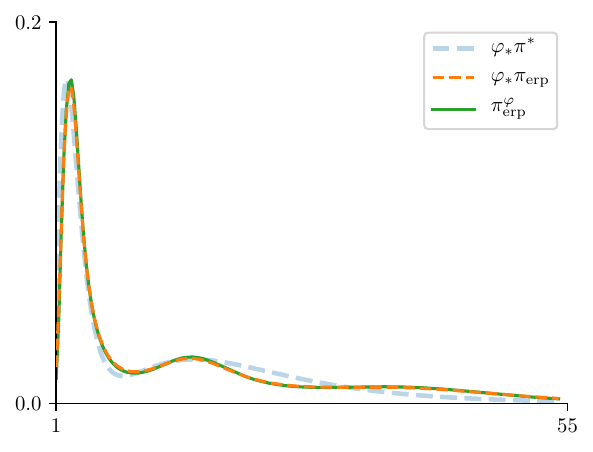}
	\end{subfigure}	
	\caption{Due to the transformation invariance of the missing information penalty $\Phi_\cI$, the empirical reference prior estimates are consistent under reparametrization (up to a negligible numerical error), $\parden_{\textup{erp}}^\varphi = \varphi_\ast\parden_{\textup{erp}}$.}
	\label{fig:piRef}	
\end{figure}
According to \eqref{equ:phiTransform}, the transformed model takes the form
\begin{equation}
\label{equ:locModelTransformed}
\meas | \tilde \param \sim \mathcal{N}(\log \tilde \param, \sigma^2 ),
\qquad
\tilde \param \sim \varphi_\ast\pi_\text{true}.
\end{equation}
Both optimization problems \eqref{equ:defMPLE} and \eqref{equ:KullbackLeiblerPenaltyMPLE} are solved by 	\Cref{algorithm:OptimizationERP} using the method of moving asymptotes (MMA) of \citet{svanberg2002class} from the NLopt package provided by \citet{nlopt}, after discretizing the parameter space with 200 equidistant grid points.
All priors are estimated using the same data consisting of $M=100$ synthetic measurements.
In both cases, the Jeffreys prior can be computed analytically. We obtain a constant Jeffreys prior for the likelihood model in \eqref{equ:locModel} and $J(\tilde \param) \propto |\tilde \param |^{-1}$ for the one in \eqref{equ:locModelTransformed}.

As expected from the theory in Section \ref{section:ObjectiveEmpirical}, we observe how the lack of invariance of conventional penalty terms is resolved by the missing information penalty $\Phi_\cI$, without losing the effect of regularization, see Figures \ref{fig:Tikhonov} and \ref{fig:piRef}.
In addition, the empirical reference prior is strictly positive on the whole interval, making $\Phi_\cI$ preferable to Tikhonov regularization from an objective Bayesian point of view: Excluding certain parameter values completely from a finite number of measurements appears unreasonable.

\subsection{Example with real-life data}
\label{section:NumericalRealLifeDataSet}
To illustrate our approach on real-life data, let us consider the historical batting averages (or baseball) example of \citet{efron1975data}, which is one of the most famous small data sets in statistics.
The batting averages $\meas_m$ (number of successful hits divided the number of tries) for $M=18$ major league baseball players early in the 1970 season --- to be precise, the first $N=45$ at bats of each --- are used to estimate their true success rates $\param_m \in \parsp = [0,1]$, $m=1,\dots,M$, which are taken as the averages over the remaining season (around 370 at bats for each player).
This example was used by \citet{efron1975data} to illustrate the strength of the James--Stein (JS) estimator $\hat{\param}_{\textup{JS}}$, a particular parametric empirical Bayes method, compared to the MLE $\param_{\MLE} = \meas$ in terms of the (empirical) mean squared error
\[
\MSE
\defeq
\frac{1}{M} \sum_{m=1}^{M} (\hat{\param}_{m} - \param_{m})^{2}.
\]
Since the JS estimator assumes both the prior and the likelihood model to be normal, the natural binomial likelihood model $\model^{\Bin}$ given by
\begin{align}
	\label{equ:BinomialModel}
	&(N\meas_{m}) | \param_{m} \stackrel{\text{independent}}{\sim} \Bin\left(\Cdot | \param_{m}, N \right),
	&&
	\param_{m} \stackrel{\text{i.i.d.}}{\sim} \parden_{\text{true}},
	&&
	m=1,\dots,M,
	\intertext{
		is replaced by its normal approximation $\model^{\cN}$,
	}
	\label{equ:NormalModel}
	&\meas_{m} | \param_{m} \stackrel{\text{independent}}{\sim} \mathcal{N}\left(\param_{m}, \sigma^2 \right),
	&&
	\param_{m} \stackrel{\text{i.i.d.}}{\sim} \parden_{\text{true}},
	&&
	m=1,\dots,M,
\end{align}
where $\sigma^{2} \defeq \overline{x}(1-\overline{x})/N$ chosen to approximate the variance estimate of the binomial distribution and $\overline{x}\defeq M^{-1}\sum_{m=1}^{M} x_m$.
Note that in the original paper by \citet{efron1975data} an additional arcsin transformation is applied to each of the measurements $x_m$ in order to stabilize their variance. For simplicity and following the presentation in \citet[Section 1.2]{efron2012large}, we omit this technical detail.
We will apply our empirical reference prior approach to both, the original formulation \eqref{equ:NormalModel} using $\model^{\cN}$ as well as the more meaningful likelihood model $\model^{\Bin}$ in \eqref{equ:BinomialModel}.
In both cases, the Jeffreys prior can be computed analytically: We obtain a constant Jeffreys prior for the likelihood model $\model^{\cN}$ in \eqref{equ:NormalModel} and $J(\param) \propto |\param (1-\param)|^{-1/2}$ for $\model^{\Bin}$ in \eqref{equ:BinomialModel}.
\begin{figure}[H]
	\centering
	\includegraphics[width=0.8\textwidth]{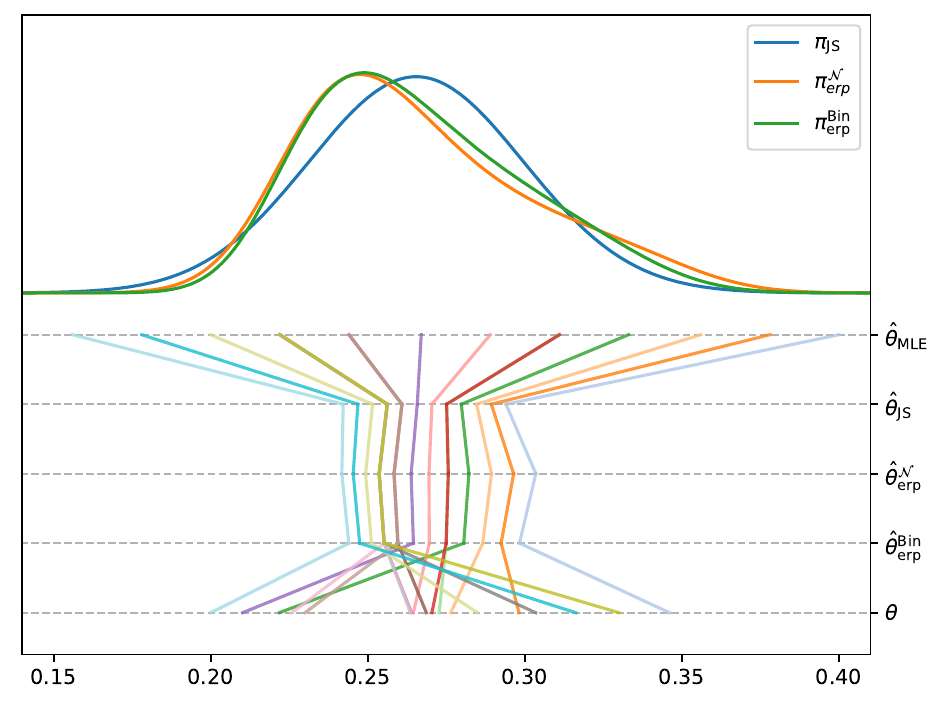}
	\caption{
		The empirical reference priors $\parden_{\textup{erp}}^{\cN}$ and $\parden_{\textup{erp}}^{\Bin}$ and the corresponding posterior mean estimates $\hat{\param}_{m}^{\textup{erp}}$ for the likelihood models given by \eqref{equ:BinomialModel} and \eqref{equ:NormalModel} compared to the maximum likelihood estimates $\hat{\param}_{m}^{\MLE}$ and James--Stein estimates $\hat{\param}_{m}^{\textup{JS}}$ (and prior $\parden_{\textup{JS}}$).
		While the results of the empirical reference prior approach and the JS approach are comparable, the normal prior $\parden_{\textup{JS}}$ clearly fails to be invariant under reparametrizations.
		Hence, the quality of the JS estimates could change drastically if, say, the players' performance was measured on a logarithmic scale.		
		\\		
		The ``shrinkage toward the mean'' effect of empirical Bayes estimates is well observable. The first impulse, that a smaller shrinkage would provide better results, is deceptive --- this would reduce the (squared) error for some players, but drastically increase the (squared) error for others, a modification that is likely to increase the overall MSE.		
	}
	\label{fig:Baseball}
\end{figure}
We compare these two approaches to the JS estimator which assumes a normal prior $\parden_{\text{true}} = \cN (\mu_{\param},\sigma_{\param}^{2})$, making it a \emph{parametric} empirical Bayes method with only the parameters $\mu_{\param}$ and $\sigma_{\param}^{2}$ of the prior to be estimated.
A wonderful derivation of the JS estimator and its application to the baseball dataset can be found in \citet[Section 1.2]{efron2012large}, with the resulting estimates
\[
\hat{\mu}_{\param}
=
\overline{x}\defeq \frac{1}{M}\sum_{m=1}^{M} x_m,
\qquad
\hat{\sigma}_{\param}^{2}
=
\frac{B\sigma^{2}}{1-B},
\qquad
\hat{\param}_{\textup{JS},m}
=
\overline{x} + B (x_m - \overline{x}),
\]
where, for $m=1,\dots,M$, the estimates $\hat{\param}_{\textup{JS},m}$ approximate $\param_{m}$ by the individual posterior means and
\[
B
\defeq
1-\frac{(M-3)\sigma^{2}}{\sum_{m=1}^{M}(x_{m} - \overline{x})^{2}}.
\]
The results are given in \Cref{table:Baseball}, while \Cref{fig:Baseball} illustrates the data as well as the various point estimates, together with the three estimated priors
\[
\parden_{\textup{JS}} = \cN (\hat{\mu}_{\param},\hat{\sigma}_{\param}^{2}),
\qquad
\parden_{\textup{erp}}^{\cN}
=
\parden_{\textup{erp}}(\Cdot\, |\, \model^{\cN}),
\qquad
\parden_{\textup{erp}}^{\Bin}
=
\parden_{\textup{erp}}(\Cdot\, |\, \model^{\Bin})
\]
\begin{table}[H]
	\centering
	\renewcommand{\arraystretch}{1.3}         
	\begin{tabular}{l x{6em}x{4em}x{4em}x{4em}x{7em}}
		\toprule
		Name & $\hat{\param}_{\MLE} = x$ & $\hat{\param}_{\textup{JS}}$ & $\hat{\param}_{\textup{erp}}^{\cN}$ & $\hat{\param}_{\textup{erp}}^{\Bin}$ & $\param$
		\\
		\midrule
		\rowcolor{black!10}
		Clemente & $0.400 = \nicefrac{18}{45}$ & $0.294 $& $0.303$ & $0.298$ & $0.346 = \nicefrac{127}{367} $\\
		Robinson & $0.378 = \nicefrac{17}{45}$ & $0.289 $& $0.296$ & $0.292$ & $0.298 = \nicefrac{127}{426} $\\
		\rowcolor{black!10}
		Howard & $0.356 = \nicefrac{16}{45}$ & $0.285 $& $0.289$ & $0.287$ & $0.276 = \nicefrac{144}{521} $\\
		Johnstone & $0.333 = \nicefrac{15}{45}$ & $0.280 $& $0.282$ & $0.281$ & $0.222 = \nicefrac{61}{275} $\\
		\rowcolor{black!10}
		Berry & $0.311 = \nicefrac{14}{45}$ & $0.275 $& $0.276$ & $0.275$ & $0.273 = \nicefrac{114}{418} $\\
		Spencer & $0.311 = \nicefrac{14}{45}$ & $0.275 $& $0.276$ & $0.275$ & $0.270 = \nicefrac{126}{466} $\\
		\rowcolor{black!10}
		Kessinger & $0.289 = \nicefrac{13}{45}$ & $0.270 $& $0.269$ & $0.270$ & $0.265 = \nicefrac{155}{586} $\\
		Alvarado & $0.267 = \nicefrac{12}{45}$ & $0.266 $& $0.264$ & $0.265$ & $0.210 = \nicefrac{29}{138} $\\
		\rowcolor{black!10}
		Santo & $0.244 = \nicefrac{11}{45}$ & $0.261 $& $0.258$ & $0.260$ & $0.269 = \nicefrac{137}{510} $\\
		Swaboda & $0.244 = \nicefrac{11}{45}$ & $0.261 $& $0.258$ & $0.260$ & $0.230 = \nicefrac{46}{200} $\\
		\rowcolor{black!10}
		Petrocelli & $0.222 = \nicefrac{10}{45}$ & $0.256 $& $0.254$ & $0.255$ & $0.264 = \nicefrac{142}{538} $\\
		Rodriguez & $0.222 = \nicefrac{10}{45}$ & $0.256 $& $0.254$ & $0.255$ & $0.226 = \nicefrac{42}{186} $\\
		\rowcolor{black!10}
		Scott & $0.222 = \nicefrac{10}{45}$ & $0.256 $& $0.254$ & $0.255$ & $0.303 = \nicefrac{132}{435} $\\
		Unser & $0.222 = \nicefrac{10}{45}$ & $0.256 $& $0.254$ & $0.255$ & $0.264 = \nicefrac{73}{277} $\\
		\rowcolor{black!10}
		Williams & $0.222 = \nicefrac{10}{45}$ & $0.256 $& $0.254$ & $0.255$ & $0.330 = \nicefrac{195}{591} $\\
		Campaneris & $0.200 = \nicefrac{9}{45}$ & $0.252 $& $0.249$ & $0.251$ & $0.285 = \nicefrac{159}{558} $\\
		\rowcolor{black!10}
		Munson & $0.178 = \nicefrac{8}{45}$ & $0.247 $& $0.245$ & $0.247$ & $0.316 = \nicefrac{129}{408} $\\
		Alvis & $0.156 = \nicefrac{7}{45}$ & $0.242 $& $0.242$ & $0.244$ & $0.200 = \nicefrac{14}{70} $\\
		\midrule
		\midrule
		{\Large $\frac{\textup{\textbf{MSE}}}{\textup{\textbf{MSE}}(\hat{\param}_{\MLE})}$} & 1 & 0.312 & 0.312 & 0.310 & (0)
		\\
		\bottomrule
	\end{tabular}
	\caption{Numerical comparison of four different estimates of the batting success rate $\param$ of 18 major league baseball players based on their batting average $x$ of the first 45 at bats early in the 1970 season ($x = \nicefrac{\text{hits}}{\text{number of at bats}}$). As illustrated in \Cref{fig:Baseball}, the three empirical Bayes methods perform very similarly, both in terms of the individual estimates $\hat{\param}_{\textup{JS}} \approx \hat{\param}_{\textup{erp}}^{\cN} \approx \hat{\param}_{\textup{erp}}^{\Bin}$ and in terms of the overall mean squared error (MSE), and strongly outperform the ``overconfident'' maximum likelihood estimate $\hat{\param}_{\MLE} = x$. The ground truth for $\param$ is taken as the batting average of each player over the remaining season.
	}
	\label{table:Baseball}
\end{table}
(note that those are \emph{not} density estimates in the classical sense, since the parameters $\param_{m}$, and not the measurements $x_{m}$, are samples from $\parden_{\text{true}}$).

All three approaches provide very similar results. The main reason for this is that, as is apparent from \Cref{fig:Baseball}, the assumption of a normal prior made by the James--Stein approach seems to be a meaningful parametric choice in this specific example, where the parameter $\param$ represents the performance (in terms of batting success rates) of baseball players. If this was not the case, e.g.\ if another parametrization was chosen by, say, measuring the performance on a logarithmic scale, our nonparametric approach is likely to outperform the James--Stein estimator.
In fact, our approach would provide \emph{exactly} the same results for any reparametrization of the model. To the best of our knowledge, this is a unique feature among all empirical Bayes methods.
\section{The multiparameter case $d>1$}
\label{section:HighDimensions}

In the case of several parameters, the reference prior $\parden_{\textup{ref}}$ is no longer defined as the prior which maximizes the missing information, but by the sequential scheme presented in \citet{berger1992development}. This scheme applies the procedure described in Section \ref{section:ObjectiveEmpirical} successively to the conditional priors $\parden(\param_\delta|\param_1,\dots,\param_{\delta-1})$, $\delta = 1,\dots,d$, after a convenient ordering of the parameters.
This leads us to three possible generalizations of empirical reference priors to the multiparameter case. The construction from Section \ref{section:ObjectiveEmpirical}, in particular the definition of the empirical reference prior, will be referred to as the \emph{one-parameter construction}.
\begin{enumerate}[label=(\Alph*)]
\item
\label{item:MultiparameterChoiceJeffreys}
Adopt Definitions from Section \ref{section:ObjectiveEmpirical} exactly as they are. In the asymptotically normal case given by Condition \ref{cond:AsymptoticNormality}, this corresponds to the optimization problem \eqref{equ:KullbackLeiblerPenaltyMPLE},
\[
\parden_{\textup{erp}} = \argmax_{\parden}\, \log L(\parden) - \gamma \dkl{\parden}{J},
\]
where $J(\param) = \sqrt{|\det i(\param)|}$ denotes the (arbitrarily scaled) Jeffreys prior. This construction provides an extension of the Jeffreys prior, \emph{not} of reference priors.
The reasons why reference priors are favored over Jeffreys prior in dimension $d>1$ are marginalization paradoxes and inconsistencies of the latter, see \citet{bernardo2005reference} and references therein. It is yet unclear in how far these arguments are valid in the presence of data. Hence, this approach might still be justified in the empirical Bayes framework.
\item
\label{item:MultiparameterChoiceReference}
In light of \eqref{equ:KullbackLeiblerPenaltyMPLE} and with the intention of generalizing reference priors, replace $J$ by $\parden_{\textup{ref}}$ in the penalty term:
\[
\parden_{\textup{erp}} = \argmax_{\parden}\, \log L(\parden) - \gamma \dkl{\parden}{\pi_{\textup{ref}}}.
\]
This yields an extension of reference priors and agrees with the one-parameter construction in the asymptotically normal case (Condition \ref{cond:AsymptoticNormality}), but not necessarily in the general case.
\item
\label{item:MultiparameterChoiceSimilarToReference}
Define a sequential scheme similar to the one used for reference priors.
For simplicity, we will restrict the presentation to the case of $d=2$ parameters $\param_1,\param_2$, $\parsp = \parsp_1\times \parsp_2$, but the construction can easily be extended to any number of parameters.
As is common practice for reference priors, the parameters have to be ordered by `inferential importance'.
We will perform similar steps to the ones in \citet[Section 3.8]{bernardo2005reference}.
Note that, as in the case of reference priors, this scheme lacks objectivity since it requires an ordering of the parameters, which is a heuristic element and not unambiguous in many applications.
\begin{algorithm}
\ 
\begin{enumerate}[label=(\roman*)]
	\item
	For every (fixed) $\param_1$, the one-parameter algorithm yields the \emph{conditional} empirical reference prior $\parden_{\textup{erp}}(\param_2|\param_1) = \parden_{\textup{erp}}(\param_2|\param_1,\model,\pardensp,\data)$.
	\item
	By integrating out parameter $\param_2$ we obtain the one-parameter model $\model_1$ given by
	\[
	p(x|\param_1) = \int_{\parsp_2} p(\meas|\param_1,\param_2)\, \parden_{\textup{erp}}(\param_2|\param_1)\, \mathrm d\param_2.
	\]
	Apply the one-parameter construction to $\model_1$ to obtain the \emph{marginal} empirical reference prior $\parden_{\textup{erp}}(\param_1) = \parden_{\textup{erp}}(\param_1|\model,\pardensp,\data)$.
	\item 
	The desired empirical reference prior is defined by $\parden_{\textup{erp}}(\param_1,\param_2) = \parden_{\textup{erp}}(\param_1) \parden_{\textup{erp}}(\param_2|\param_1)$.
\end{enumerate}
\end{algorithm}
\end{enumerate}
(A) and (B) are straightforward generalizations of the theory presented in Section \ref{section:ObjectiveEmpirical}. It would be interesting to analyze the connection between the approaches (B) and (C).
\begin{remark}
	Note that while all three of the above options are invariant under transformations of the measurement $\meas$, only \ref{item:MultiparameterChoiceJeffreys} truly guarantees invariance under transformations of the parameter $\param = (\param_{1},\dots,\param_{d})$, which it inherits from the invariance of the log-likelihood (up to an additive constant, \Cref{theorem:invarianceLandI}) and of the Jeffreys prior.
	The other two options \ref{item:MultiparameterChoiceReference} and \ref{item:MultiparameterChoiceSimilarToReference} depend on the particular choice of how the parameters (or, rather, the parameter components) $\param_{1},\dots,\param_{d}$ are ordered. Hence invariance is only given under reparametrizations that transform each component separately,
	\[
	\varphi(\param) = \big(\varphi_{1}(\param_{1}),\dots,\varphi_{d}(\param_{d})\big),
	\]
	under the assumptions that their ordering is not altered.
\end{remark}
\section{Conclusion}
\label{section:Conclusion}

We successfully applied the approach for the construction of reference priors to determine a transformation invariant penalty term for MPLE, which favors non-informativity of the prior, namely the missing information $\cI[\parden\, |\, \model^\replications]$, $\replications\to\infty$.
This interaction of objective Bayesian analysis and empirical Bayes methods results in a consistent and informative prior estimate, which we termed the empirical reference prior $\parden_{\textup{erp}}$.
The distinctive feature of $\parden_{\textup{erp}}$ is its invariance under reparametrization, which is, to the best of our knowledge, unique among all empirical Bayes methodologies.

The smoothing parameter $\gamma$ tunes the amount of information contained in the prior:
The data, represented by the marginal likelihood $L(\parden)$, yields information about the distribution of $\param$, but maximizing $L(\parden)$ alone overfits the data. The penalty term $\Phi(\pi)$, on the other hand, favors non-informative priors. We performed this trade-off by likelihood cross-validation which we also proved to be invariant under transformations (\Cref{equ:InvarianceSmoothingParameter}).

Besides invariance, our method has further favorable properties such as compatibility with sufficient statistics and, under the assumption of asymptotic normality, strict concavity of the resulting optimization problem \eqref{equ:KullbackLeiblerPenaltyMPLE}. So far, our approach lacks an explicit formula for the empirical reference prior.\footnote{An explicit formula for reference priors \citep[Theorem 7]{berger2009formal} exists, so far, only in the one-parameter case $d=1$ and under rather restrictive assumptions.
Further, it requires the numerical approximations of integrals in the possibly high-dimensional measurement space $\meassp$ as well as the computation or at least some estimate of the limit $\replications\to\infty$.}

We applied our methodology to a synthetic data set to illustrate the invariance property of the empirical reference prior as well as to a real-life example, the famous baseball data set of \citet{efron1975data} collected to demonstrate the advantages of the James--Stein estimator, a parametric empirical Bayes method which we compare our method to.
In the latter case, both approaches gave nearly identical results, the most natural explanation being that the parametric choice of the prior used in the James--Stein approach appears to be a good approximation in this specific example with this specific parametrization.
Our nonparametric empirical reference prior approach is likely to provide better results in situations where the ``true'' prior $\parden_{\text{true}}$ can not be well-approximated by a Gaussian.
Our code for the numerical computations is available at \url{https://github.com/axsk/ObjectiveEmpiricalBayes.jl}.

The generalization of our approach to several dimensions is not unambiguous and has been discussed in Section~\ref{section:HighDimensions}.

\section*{Acknowledgements}
\addcontentsline{toc}{section}{Acknowledgements}

We thank Ingmar Schuster for countless enlightening discussions on the topic.

\bibliographystyle{abbrvnat}
\bibliography{myBibliography}
\addcontentsline{toc}{section}{References}

\end{document}